\documentclass[runningheads]{llncs}

\usepackage{microtype}
\usepackage{amsmath}
\usepackage{graphicx}
\usepackage{xcolor}
\usepackage{soul}
\usepackage{algorithmic}
\usepackage{algorithm}
\usepackage{tabto}
\newcommand{\remove}[1]{}
\usepackage{graphicx}
\begin{document}

\newcommand{\UKMPI}{(\clientset, ~\facilityset,~\textit{c},~\textit{f},~\capacity, ~\budget)}			
\newcommand{\InstanceUKMP}{(\clientset, ~\facilityset,~\textit{c},~\textit{f},~\capacity, ~\budget)}	

\newcommand{\UkFLPI}{(\clientset,~\facilityset,~\textit{c},~\textit{f},~\capacity,~k)}					
\newcommand{\InstanceUKFLP}{(\clientset,~\facilityset,~\textit{c},~\textit{f},~\capacity,~k)}			
\newcommand{\UFLPI}{(\clientset,~\facilityset,~\textit{c},~\textit{f},~\capacity)}	

\newcommand{\starinst}[5]{\mathcal{#1}_{#2}(#2,~#3,~#4,~#5)}	

\newcommand{\starinstflp}[4]{\mathcal{#1}_#2(#2,~#3,~#4)}												
\newcommand{\InstanceSTAR}[5]{\mathcal{#1}_#2(#2,~#3,~#4,~#5)}											



\newcommand{\overbar}[1]{\mkern 1.5mu\overline{\mkern-1.5mu#1\mkern-1.5mu}\mkern 1.5mu}			

\newcommand{\thetajjp}[2]{\theta({#1},{#2})}


\newcommand{\LPUKMP}{\textbf{LP(UKMP)}}			


\newcommand{\facilityset}{\mathcal{F}}	
\newcommand{\facilitysetone}{\mathcal{F}^1}	
\newcommand{\facilitysettwo}{\mathcal{F}^2}		
\newcommand{\facset}{\mathcal{F}}				
\newcommand{\opfacset}{\mathcal{A}}	
\newcommand{\opfacsettwo}{\mathcal{A}^2}	
\newcommand{\temp}{\mathcal{T}}	

\newcommand{\clientset}{\mathcal{C}}			
\newcommand{\cliset}{\mathcal{C}}				

\newcommand{\si}[1]{\mathcal{S}_#1}				
\newcommand{\sset}[1]{\mathcal{S}_#1}			

\newcommand{\budget}{\mathcal{B}}
\newcommand{\budgetone}{\mathcal{B}^1}
\newcommand{\budgettwo}{\mathcal{B}^2}				
\newcommand{\budgetset}[1]{\mathcal{B}_{#1}}

\newcommand{\capacity}{\textit{u}}				
\newcommand{\capacityU}{\textit{U}}				

\newcommand{\crich}{\mathcal{C}_D}				
\newcommand{\cpoor}{\mathcal{C}_S}				
\newcommand{\cdense}{\mathcal{C}_D}				
\newcommand{\csparse}{\mathcal{C}_S}
\newcommand{\csparseprime}{\mathcal{C}'_S}
\newcommand{\csparsetilde}{\mathcal{\tilde{C}}_S}
\newcommand{\csparsegroup}[1]{\mathcal{C}_{S,#1}}
\newcommand{\facilitysetgroup}[1]{\mathcal{F}_{#1}}
\newcommand{\csone}{\mathcal{C}_S^1}		
\newcommand{\cstwo}{\mathcal{C}_S^2}
\newcommand{\cdone}{\mathcal{C}_D^1}	

\newcommand{\cdensegood}{\mathcal{C}^D_{Good}}	
\newcommand{\cdensebad}{\mathcal{C}^D_{Bad}}	
\newcommand{\cminusl}{\mathcal{\tilde{C}}}	
\newcommand{\csparseminusl}{\mathcal{\tilde{C}_S}}

\newcommand{\cdensegoodname}{\textit{Good Dense}}	
\newcommand{\cdensebadname}{\textit{Bad Dense}}		
\newcommand{\csparsename}{\textit{Sparse}}
\newcommand{\cdensename}{\textit{Dense}}

\newcommand{\csparseplusbad}{\overline{\mathcal{C}}}		
\newcommand{\leafnodes}{\mathcal{L}}		



\newcommand{\ballofj}[1]{\textit{\mathcal{B}($ #1 $)}} 	

\newcommand{\C}[1]{\hat{C_{#1}}}					
\newcommand{\hatofC}[1]{\hat{C_{#1}}}				
\newcommand{\hatof}[2]{\hat{{#1}_{#2}}}				

\newcommand{\bundle}[1]{\mathcal{N}_{#1}}			
\newcommand{\neighbor}[1]{\mathcal{N}_{#1}}
\newcommand{\p}[1]{\mathcal{N}_{#1}}
\newcommand{\nnew}[1]{\mathcal{N}'_{#1}}
\newcommand{\Njincbar}[1]{\delta({#1})}				
\newcommand{\Tp}[1]{\mathcal{T}'({#1})}
\newcommand{\TC}{\mathcal{TC}}

\newcommand{\cen}[1]{G_{#1}}
\newcommand{\sj}[1]{\mathcal{S}_{#1}}
\newcommand{\s}{\mathcal{S}}	

\newcommand{\dist}[2]{c(#1,~#2)}
\newcommand{\dists}[2]{c_{s}(#1,~#2)}					

\newcommand{\facilitycost}{\textit{$ f_i $}}		
\newcommand{\faccosti}{\textit{$ f_i $}}			
\newcommand{\faccost}[1]{\textit{$ f_#1 $}}			

\newcommand{\A}[2]{\mathcal{A}_{\rho^*}(#1,{#2})}						

\newcommand{\x}[3]{x_{#1}(#2,#3)}					
\newcommand{\X}[2]{x_{#1#2}}						
\newcommand{\Xstar}[2]{x^*_{#1#2}}					

\newcommand{\bard}[1]{{d_{#1}}}			

\newcommand{\hatofbjc}[1]{\hat{b}^c_{#1}}


\newcommand{\sumlimits}[2]{\displaystyle\sum\limits_{#1}^{#2}}												
\newcommand{\unionlimits}[2]{\displaystyle\bigcup\limits_{#1}^{#2}}

\newcommand{\sumofx}[1]{\sum_{#1 \in \clientset} x_{i#1}}													
\newcommand{\sumxoverclients}[2]{\displaystyle\sum\limits_{#2 \in \clientset} x_{#1#2}}						

\newcommand{\sumxioverclients}[1]{\displaystyle\sum\limits_{#1 \in \clientset} x_{i#1}}						
\newcommand{\sumofvaronf}[2]{\sum_{i \in \facilityset'} #1_{#2j}}											

\newcommand{\sumap}[1]{\sum_{#1}{}} 


\newcommand{\price}[1]{b^c_{#1}}
\newcommand{\pricef}[1]{b^f_{#1}}
\newcommand{\budgett}[1]{b_{#1}}

\newcommand{\load}[1]{l_{#1}}

\newcommand{\mass}[2]{size(#1,#2)}
\newcommand{\Load}[1]{load(#1)}
\newcommand{\massofbundle}[3]{\sum_{#1 \in #2}{} #3_#1}
\newcommand{\FC}[1]{\sum_{\singlefacility \in \facilityset} \facilitycost  #1_i}
\newcommand{\FCofbundle}[3]{\sum_{#1 \in #2} f_#1 #3_#1}
\newcommand{\CC}[1]{\sum_{\singlefacility \in \facilityset} \sum_{j \in \clientset} \dist{i}{j} #1_{ij}}
\newcommand{\CCa}{\sum_{i,j} \demandofj{j} \dist{i}{j}}

\newcommand{\CCofBundle}[1]{\sum_{i \in \facilityset} \dist{#1}{j} x^*_{#1j}}
\newcommand{\kmed}{\textit{k}}     

\newcommand{\singleclient}{\textit{j}}
\newcommand{\singlefacility}{\textit{i}}
\newcommand{\demand}[1]{\tilde{l_{#1}}}
\newcommand{\hatofdemand}[1]{\hat{l_{#1}}}
\newcommand{\zofi}[1]{z_#1}
\newcommand{\hatzofi}[1]{\hat{z}_#1}
\newcommand{\primezofi}[1]{z'_{#1}}
\newcommand{\tildezofi}[1]{\tilde{z}_#1}

\newcommand{\clientsetbundle}[1]{\mathcal{C}^{#1}}
\newcommand{\radiusballofj}{\dist{i}{j} \leq 2 \C{j} }
\newcommand{\demandofj}[1]{\textit{$ d_{#1} $}}
\newcommand{\neighbourhood}[1]{\mathbb{N}(#1)}
\newcommand{\shtree}{\cal{F}}
\newcommand{\BCtree}{\hat{\cal{F}}}
\newcommand{\CCLPtwo}[1]{\sum_{i \in \facilityset',j \in \clientsetbundle{'}} \demandofj{j} \dist{i}{j}  #1_{ij}}
\newcommand{\priceofcorf}[2]{b_{#1}^{#2}}
\newcommand{\hatofy}[1]{\hat{y_#1}}
\newcommand{\hatofx}[2]{\hat{x}_{#1#2}}
\newcommand{\barofclientsetbundle}{\bar{\mathcal{C}}}

\newcommand{\obj}[1]{\mathcal{C}ostKM(#1)} 
\newcommand{\objflp}[1]{\mathcal{C}ostkFLP_{sp}(#1)} 
\newcommand{\objflpSD}[1]{\mathcal{C}ostkFLP(#1)} 

\newcommand{\Costone}[1]{Cost^1(#1)}
\newcommand{\Costtwo}[1]{Cost^2(#1)}
\newcommand{\cost}[2]{Cost_{#1}(#2)} 
\newcommand{\sigmabar}{\bar{\sigma}}
\newcommand{\sigmahat}{\hat{\sigma}}
\newcommand{\costsingle}[1]{Cost(#1)}
\newcommand{\costjprimeinc}[3]{Cost_{#1}(#2,~#3)}
\newcommand{\loadjinc}[2]{\phi(#1,~#2)}

\newcommand{\costFLPThree}[3]{Cost_{#1}(#2,~#3)}
\newcommand{\costtwo}[2]{Cost(#1,~#2)}
\newcommand{\costthree}[3]{Cost_{#1}(#2,~#3)}
\newcommand{\costFLPFour}[4]{Cost_{#1}(#2,~#3,~#4)}
\newcommand{\cardTwo}[1]{Y^*(#1)}
\newcommand{\totalcost}[2]{Cost(#1,~#2)}
\newcommand{\etal}{\textit{et al}.}
\newcommand{\ie}{\textit{i}.\textit{e}.}
\newcommand{\eg}{\textit{e}.\textit{g}.}
\newcommand{\Cost}[1]{Cost(#1)}

\newcommand{\soner}[1]{s^1_{#1}}
\newcommand{\stwor}[1]{s^2_{#1}}
\newcommand{\Soner}[1]{S^1_{#1}}
\newcommand{\Stwor}[1]{S^2_{#1}}

\newcommand{\mcone}[1]{G^1_{#1}}
\newcommand{\mctwo}[1]{G^2_{#1}}
\newcommand{\denser}[1]{j_{d_{#1}}}
\newcommand{\dense}{j_d}
\newcommand{\sparseoner}[1]{j_{s1_{#1}}}
\newcommand{\sparseone}{j_{s}}
\newcommand{\sparsetwor}[1]{j_{s2_{#1}}}
\newcommand{\sparsetwo}{j_{s2}}
\newcommand{\That}[1]{\hat{\tau}({#1})}
\newcommand{\muj}[1]{\mu({#1})}
\newcommand{\resj}[1]{res({#1})}
\newcommand{\floordjbyu}[1]{\lfloor{d_{#1}/u}\rfloor}

\newcommand{\tauhat}[1]{\hat{\tau}({#1})}

\newcommand{\sone}{S_1}
\newcommand{\stwo}{S_2}
\newcommand{\sthree}{S_3}
\newcommand{\Gralpha}{G_{r_\alpha}}
\newcommand{\Grbeta}{G_{r_\beta}}

\newcommand{\sigmaone}{\psi}
\newcommand{\sr}[1]{s_{#1}}
\newcommand{\Sr}[1]{S_{#1}}

\newcommand{\B}{B}

\newcommand{\clirem}{\mathcal{C}_r}
\newcommand{\clipen}{\mathcal{C}_p}
\newcommand{\cliout}{\mathcal{C}_o}
\newcommand{\opt}[1]{LP_{opt}}
\newcommand{\lp}[1]{LP_{#1}}
\newcommand{\scost}{Cost_s}
\newcommand{\pcost}{Cost_p}
\newcommand{\fcost}{Cost_f}

\newcommand{\MC}{Meta-Cluster }
\newcommand{\MCs}{Meta-Clusters }

\newcommand{\starj}{\facilityset_j}
\newcommand{\starjp}{\facilityset_{j'}}
\newcommand{\distd}[2]{c'(#1,~#2)}
\newcommand{\cfull}{\cliset_{full}}
\newcommand{\cpart}{\cliset_{part}}
\newcommand{\cstar}{\cliset^*}
\newcommand{\ballj}{\mathcal{B}_j}
\newcommand{\dljp}{\mathcal{D}_{\ell_{j'}}}
\newcommand{\Tau}{\mathcal{T}}

\newcommand{\dlj}[1]{r\mathcal{F}_{j_{#1}}}

\newcommand{\rFj}[1]{rF_{j}_{#1}}
\newcommand{\level}[1]{\rFj_#1}
\newcommand{\rhobar}{\bar\rho}
\newcommand{\lb}[1]{\mathcal{L}_{#1}}
\newcommand{\AVG}[1]{\mathcal{A}_{#1}}
\newcommand{\radj}[1]{\mathcal{R}_{#1}}
\newcommand{\cluster}[1]{\mathcal{P}_{#1}}
\newcommand{\F}[1]{\mathcal{F}_{#1}}
\newcommand{\iofj}{i_j}
\newcommand{\N}[1]{\mathcal{N}_{#1}}
\newcommand{\iclients}{\clientset_i}
\newcommand{\Out}[1]{\mathcal{O}_{#1}}
\newcommand{\R}[1]{\mathcal{R}_{#1}}
\newcommand{\rjmax}{rmax_j}
\newcommand{\T}[1]{\mathcal{T}_{#1}}
\newcommand{\cli}[1]{j_{}}
\newcommand{\ball}[1]{\mathcal{B}_{#1}}
\newcommand{\floor}[1]{\lfloor{#1}\rfloor}
\newcommand{\rtj}[1]{r\mathcal{T}_{j_{#1}}}

\title{On Variants of Facility Location Problem with Outliers}

\titlerunning{Variants of Facility Location Problem with Outliers}

\author{Rajni Dabas\inst{1} \and
Neelima Gupta\inst{1}}

\authorrunning{R. Dabas and N. Gupta}

\institute{Department of Computer Science, University of Delhi, India\\
\email{rajni@cs.du.ac.in, ngupta@cs.du.ac.in}}

\maketitle             
\begin{abstract}
In this work, we study the extension of two variants of the facility location problem (FL) to make them robust towards a few distantly located clients. First, $k$-facility location problem ($k$FL), a common generalization of  FL  and $k$ median problems, is a well studied problem in literature. In the second variant, lower bounded facility location (LBFL),  we are given a bound on the minimum number of clients that an opened facility must serve. Lower bounds are required in many applications like profitability in commerce and load balancing in transportation problem. In both the cases, the cost of the solution may be increased grossly by a few distantly located clients, called the outliers. Thus, in this work, we extend $k$FL and LBFL to make them robust towards the outliers.
For $k$FL with outliers ($k$FLO) we present the first (constant) factor approximation violating the cardinality requirement by +1. As a by-product, we also obtain the first approximation for FLO based on LP-rounding. For LBFLO, we present a tri-criteria solution with a trade-off between the violations in lower bounds and the number of outliers. With a violation of $1/2$ in lower bounds, we get a violation of $2$ in outliers.

\keywords{ Facility Location  \and Outliers \and Approximation \and Lower Bound \and $k$-Facility Location \and $k$-Median.}
\end{abstract}

\section{Introduction}
Consider an e-retail company that wants to open warehouses in a city for home delivery of essential items. Each store has an associated opening cost depending on the location in the city. The aim of the company is to open these warehouses at locations such that the cost of opening the warehouses plus the cost servicing all the customers in the city from the nearest opened store is minimised. In literature, such problems are called {\em facility location problems}(FL) where warehouses are the facilities and customers are the clients. Formally, in FL we are given a set $\facilityset$ of $n$ facilities and a set $\cliset$ of $m$ clients. Each facility $i \in \facilityset$ has an opening cost $f_i$ and cost of servicing a client $j \in \cliset$ from a facility $i \in \facilityset$ is $c(i,j)$ (we assume that the service costs are metric). The goal is to open a subset $\facilityset' \subseteq \facilityset$ of facilities such that the cost of opening the facilities and servicing the clients from the opened facilities is minimised.
In a variant of FL, called $k$-facility location problem ($k$FL), we are  given an additional bound $k$ on the maximum number of warehouses/facilities that can be opened i.e.  $|\facilityset'| \le k$. In our example this requirement may be imposed to maintain the budget constraints or to comply with government regulations. In another variant of the problem, we are required to serve some minimum number of customers/clients from an opened facility. Such a requirement is natural to ensure profitability in our example.  
This minimum requirement is captured as lower bounds in facility location problems. That is, in lower bounded FL (LBFL), we are also given a lower bound $\lb{i}$ on the minimum number of clients that an opened facility $i$ must serve.

In the above scenarios, a few distant customers/clients can increase the cost of the solution disproportionately; such clients are called {\em outliers}.
Problem of outliers was first introduced by Charikar \etal~\cite{charikar2001algorithms} for the facility location and the $k$-median problems. In this paper we extend $k$-facility location and lower bounded facility location to deal with the outliers and denote them by $k$FLO and LBFLO respectively. Since FL is well known to be NP-hard, NP-hardness of $k$FLO and LBFLO follows. We present the first (constant factor) approximation for $k$FLO  opening at most $k+1$ facilities. In particular, we present the following result:

\begin{theorem}
\label{thm_flo}
There is a polynomial time algorithm that approximates $k$-facility location problem with outliers opening at most $(k+1)$ facilities within $11$ times the cost of the optimal solution.
\end{theorem}

Our result is obtained using LP rounding techniques. As a by product, we get first constant factor approximation for FLO using LP rounding techniques. FLO is shown to have an unbounded integrality gap~\cite{charikar2001algorithms} with solution to the standard LP. We get around this difficulty by guessing the most expensive  facility opened in the optimal solution. 
In particular we get the following:

\begin{corollary}
\label{coro-flo}
There is a polynomial time algorithm that approximates facility location problem with outliers within $11$ times the cost of the optimal solution.
\end{corollary}



We reduce LBFLO  to FLO and use any algorithm to approximate FLO to obtain a tri-criteria solution for the problem. 
To the best of our knowledge, no result is known for LBFLO in literature. In particular, we present our result in Theorem~\ref{thm_LBkFLO} where a tri-criteria solution is defined as follows:
 
 \begin{definition}
A tri-criteria solution for LBFLO is an $(\alpha, \beta, \gamma)$- approximation solution $S$ that violates lower bounds by a factor of $\alpha$ and outliers by a factor of $\beta$ with cost no more than $\gamma OPT$ where $OPT$ denotes the cost of an optimal solution of the problem, $\alpha<1$ and $\beta>1$.
\end{definition}

\begin{theorem}
\label{thm_LBkFLO}
A polynomial time $(\alpha,\frac{1}{1-\alpha}, \lambda(\frac{1+\alpha}{1-\alpha})$-approximation can be obtained for LBFLO problem where $\alpha=(0,1)$ is a constant and $\lambda$ is an approximation factor for the FLO problem.
\end{theorem}
Theorem~\ref{thm_LBkFLO} presents a trade-off between the violations in the lower bounds and that in the number of outliers. Violation in outliers can be made arbitrarily small by choosing $\alpha$ close to $0$. And, violation in lower bounds can be chosen close to $1$ at the cost of increased violation in the outliers. Similar result can be obtained for LB$k$FLO with $+1$ violation in cardinality using Theorem~\ref{thm_flo}. The violation in the cardinality comes from that in $k$FLO.

\textbf{Our Techniques: }
For $k$FLO, starting with an LP solution $\rho^* = <x^*, y^*>$, we first eliminate the $x^*_{ij}$ variables and work with an auxilliary linear programming (ALP) relaxation involving only $y_i$ variables. This is achieved by converting $\rho^*$ into a  {\em complete solution} in which either $x^*_{ij}=y^*_{i}$ or $x^*_{ij}=0$. Using the ALP, we identify the set of facilities to open in our solution. ALP is solved using iterative rounding technique to give a pseudo-integral solution (a solution is said to be pseudo integral if there are at most two fractional facilities). We open both the facilities at $+1$ loss in cardinality and at a loss of factor $2$ in the cost by guessing the maximum opening cost of a facility in the optimal. Once we identify the set of facilities to open, we can greedily assign the first $m-t$ clients in the increasing order of distance from the nearest opened facility. Thus, in the rest of the paper, we only focus on identifying the set of facilities to open.
    
For LBFLO, we construct an instance $I'$ of FLO by ignoring the lower bounds and defining new facility opening cost for each $i \in \facilityset$. An approximate solution $AS'$ to $I'$ is obtained using any approximation algorithm for FLO. Facilities serving less than $\alpha\lb{i}$ clients are closed and their clients are either reassigned to the other opened facilities or are made outliers. This leads to violation in outliers that is bounded by $\frac{1}{1-\alpha}$. Facility opening costs in $I'$ are defined to capture the cost of reassignments.

\textbf{Related Work:} 
The problems of facility location and $k$-median with outliers were first defined by Charikar~\etal~\cite{charikar2001algorithms}.
Both the problems were shown to have unbounded integrality gap~\cite{charikar2001algorithms} with their standard LPs.
For FLO, they gave a $(3+\epsilon)$-approximation using primal dual technique by guessing the most expensive facility opened by the optimal solution. 
For a special case of the problem with uniform facility opening costs and doubling metrics, Friggstad~\etal~\cite{FriggstadKR19}
gave a PTAS using multiswap local search. 
For $k$MO, Charikar~\etal~\cite{charikar2001algorithms} gave a $4(1+1/\epsilon)$-approximation with $(1+\epsilon)$-factor violation in outliers. Using local search techniques, Friggstad~\etal~\cite{FriggstadKR19} gave $(3+\epsilon)$ and $(1+\epsilon)$-approximations with ($1+\epsilon$) violation in cardinality for general and doubling metric respectively. Chen~\cite{chen-kMO} gave the first true constant factor approximation for the problem using a combination of local search and primal dual. Their approximation factor is large and it was improved to $(7.081+\epsilon)$ by  Krishnaswamy~\etal~\cite{krishnaswamy-kMO} by strengthening the LP. They use iterative rounding framework and, their factor is the current best result for the problem.

 Lower bounds in FL were introduced by Karger and Minkoff~\cite{Minkoff_LBFL} and Guha~\etal~\cite{Guha_LBFL}. They independently gave constant factor approximations with violation in lower bounds. The first true constant factor($448$) approximation was given by Zoya Svitkina~\cite{Zoya_LBFL} for uniform lower bounds. The factor was improved to $82.6$ by Ahmadian and Swamy \cite{Ahmadian_LBFL}.  Shi Li~\cite{Li_NonUnifLBFL} gave the first constant factor approximation for general lower bounds, with the constant being large ($4000$). 
Han~\etal~\cite{Han_LBkM} studied the general lower bounded $k$-facility location  (LB$k$FL) 
violating the lower bounds.
 Same authors~\cite{Han_LBknapsackM} removed the violation in the lower bound for the $k$-Median problem.
 
 The only work that deals with lower bound and outliers together is by Ahmadian and Swamy~\cite{ahmadian_lboutliers}. They have given constant factor approximation for lower-bounded min-sum-of-radii with outliers and lower-bounded k-supplier with outliers problems using primal-dual technique.
 

\textbf{Organisation of the paper:} A constant factor approximation for $k$FLO  is given in Section~\ref{kFLPO} opening at most $(k+1)$ facilities. In Section~\ref{LBkFLO}, the tri-criteria solution for LBFLO is presented. Finally we conclude with future scope in Section~\ref{conclusion}.
\section{$(k+1)$ solution for $k$FLO}
\label{kFLPO}
The problem $k$FLO can be represented as the following integer program (IP):

\label{{k-FLPO}}
$Minimize ~\mathcal{C}ostkFLO(x,y) = \sum_{j \in \cliset}\sum_{i \in \facilityset}\dist{i}{j}x_{ij} + \sum_{i \in \facilityset}f_iy_i $
\begin{eqnarray}
subject~ to &\sum_{i \in \facilityset}{} x_{ij} \leq 1 & \forall ~\singleclient \in \clientset \label{LPFLP_const1}\\ 
& x_{ij} \leq y_i & \forall~ \singlefacility \in \facilityset , ~\singleclient \in \clientset \label{LPFLP_const2}\\  
& \sum_{i \in \facilityset}y_{i} \leq k & \label{LPFLP_const3}\\ 
& \sum_{j \in \clientset} \sum_{i \in \facilityset}x_{ij} \geq m-t &  \label{LPFLP_const4}\\   
& y_i,x_{ij} \in \left\lbrace 0,1 \right\rbrace  \label{LPFLP_const5}
\end{eqnarray}
where variable $y_i$ denotes whether facility $i$ is open or not and $x_{ij}$ indicates if client $j$ is served by facility $i$ or not. 
Constraints \ref{LPFLP_const1} ensure that the extent to which a client is served is no more than $1$.
Constraints \ref{LPFLP_const2} ensure that a client is assigned only to an open facility. Constraint \ref{LPFLP_const3} ensures that the total number of facilities opened are atmost $k$  and Constraint \ref{LPFLP_const4} ensures that total number of clients served are at least $m-t$. LP-Relaxation of the problem is obtained by allowing the variables $y_i, x_{ij} \in [0, 1]$. Let us call it $LP$.


Let $\rho^{*} = <x^*, y^*>$ denote the optimal solution of $LP$ and $\opt{}$ denote the cost of $\rho^*$. A solution is said to be a {\em complete solution} either $x^*_{ij}= y^*_i$ or $x^*_{ij}=0$, $\forall i \in \facilityset$ and $\forall j \in \clientset$. We first eliminate $x$ variables from our solution $\rho^*$ by making it complete. This is achieved by standard technique of splitting the openings and making collocated copies of facilities. For every client $j \in \cliset$, we will define a bundle, $\starj$ as the set of facilities that are serving $j$ in our complete solution. Formally, $\starj = \{ i \in \facilityset : x^*_{ij}>0\}$. Let $\dlj{} = max_{i \in \starj}\distd{i}{j}$ be the distance of farthest facility in $\starj$ from $j$. See Fig.~\ref{FIG_FLO1}($a$). Note that the complete solution $<x^*, y^*>$ satisfies the following property:
\begin{enumerate}
    \item \label{prop1} $\sum_{i \in \starj} y^*_i \leq 1~\forall j \in \cliset$ as $\sum_{i \in \starj} y^*_i = \sum_{i \in \facilityset}x^*_{ij} \leq 1$.
    \item \label{prop2} $\sum_{i \in \facilityset}y^*_i \leq k$
    \item \label{prop3} $\sum_{j \in \cliset} \sum_{i \in \starj} y^*_i \geq m-t$ as $\sum_{i \in \starj} y^*_i = \sum_{i \in \facilityset}x^*_{ij}$ and $\sum_{j \in \cliset} \sum_{i \in \facilityset} x^*_{ij} \geq m-t$.
\end{enumerate}

\subsection{Auxiliary LP (ALP)}
\begin{figure}[t]
    \begin{center}
	\begin{tabular}{cc}
		\includegraphics[width=4cm]{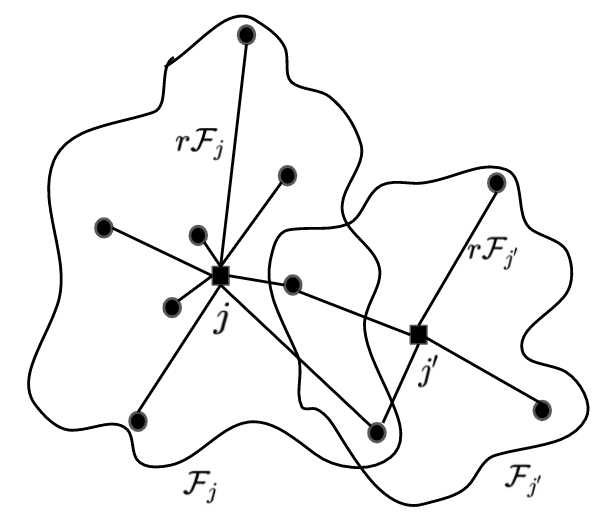}
		& 
		\includegraphics[width=4cm]{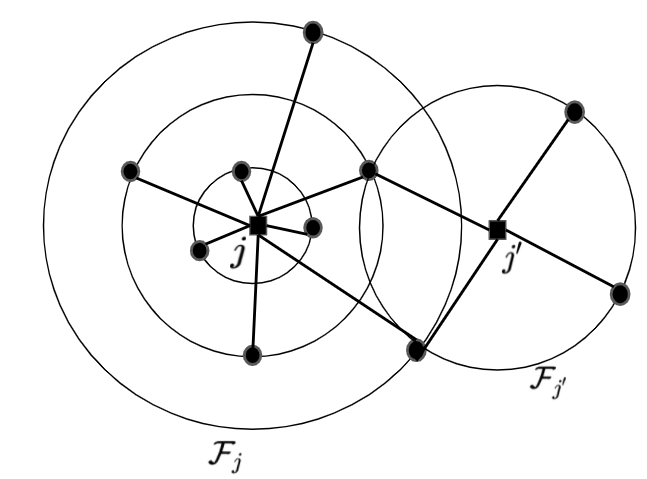}
		\\ ($a$)& \ \ ($b$) \ \
\end{tabular}
\end{center}
	\caption{ ($a$) Set $\starj$ corresponding to a client $j$, ($b$) Discretization of distances}.
	\label{FIG_FLO1}
\end{figure}

We first discretize our distances $c(i,j)$, by rounding them to the nearest power of $2$. Let $\distd{i}{j} = 2^r$, where $r$ is smallest power of $2$ such that $\dist{i}{j} \le 2^r$. See Fig.~\ref{FIG_FLO1}($b$). 
   Next, we identify a set $\cfull $ of clients that are going to be served fully in our solution. Ideally, we would like to open at least one facility in $\starj$ for every $j \in \cfull$. If all the $\starj$'s ($j \in \cfull$) were pair-wise disjoint, an LP constraint like $\sum_{i \in \starj} w^*_i \geq 1$ for all $j \in \cfull$, along with constraints~\ref{LPALP_const3}(for partially served clients, say clients in $\cpart$),~\ref{LPALP_const0}(for cardinality) and~\ref{LPALP_const4}(for outliers), is sufficient to get us a psuedo-integral solution. But this, in general, is not true. 
   Thus we further identify a set $\cstar \subseteq \cfull$ so that we open  one facility in $\starj$ for every $j \in \cstar$ and
   ($i$)  $\starj$'s ($j \in \cstar$) are pair-wise disjoint \textit{(disjointness property)}
   ($ii$) for every $\cli{f} \in \cfull \setminus \cstar$, there is a close-by (within constant factor of $\dlj{}$ distance from $\cli{f}$) client in $\cstar$. 
   On a close observation, we notice that instead of $ \starj$'s,  we are rather interested in smaller sets: let $\rjmax$ be the (rounded) distance of the farthest facility  in $\starj$ serving $j$ in our solution and $\T{j} = \{ i \in \starj : c'(i,j) \leq \rjmax\}$. Then we actually want $\T{j}$'s ($j \in \cstar$) to be pair-wise disjoint. As the distances are discretized, we have that $\rjmax$ is either $\dlj{}$ or is $\le \dlj{}/2$. Since we don't know $\rjmax$, once a client is identified to be in $\cfull$, we search for it by starting with $\T{j} = \starj$, $\rtj{} = \dlj{}$ and, shrinking it over iterations. Shrinking is done whenever, for $\ballj = \{ i \in \starj : c'(i,j) \leq \rjmax/2\}$ , we obtain $\sum_{i \in \ballj} w_i = 1$. Thus we add a constraint $\sum_{i \in \ballj} w_i \le 1$ in our ALP and arrive at the following auxiliary LP (ALP). Variable $w_i$ denotes whether facility $i$ is opened in the solution or not. Constraints (9) and (10) correspond to the requirements of cardinality and outliers. For $\cli{f}\in \cfull$, if the ALP doesn't open a facility within $\ball{\cli{f}}$, it bounds the cost of sending $\cli{f}$ up to a distance of $\rtj{}$.

$\text{Min}~CostALP(w) = \sum_{j \in \cpart}\sum_{i \in \T{j}}\distd{i}{j} w_i + \sum_{j \in \cfull} [ \sum_{i \in \ballj} \distd{i}{j} w_i + (1 - \sum_{i \in \ballj} w_i)\rtj{}] + \sum_{i \in \facilityset}f_i w_i$
\begin{eqnarray}
subject~ to 
&\sum_{i \in \T{j}} w_i = 1 & \forall ~ j \in \cstar \label{LPALP_const1}\\
&\sum_{i \in \ballj} w_i \leq 1 & \forall ~ j \in \cfull \label{LPALP_const2}\\
&\sum_{i \in \T{j}} w_i \leq 1 & \forall ~ j \in \cpart \label{LPALP_const3}\\
&\sum_{i \in \facilityset} w_i \leq k & \label{LPALP_const0}\\
&|\cfull| + \sum_{j \in \cpart} \sum_{i \in \T{j}} w_i \geq m-t & \label{LPALP_const4}\\   
& 0 \leq w_i \leq 1  &\label{LPALP_const5}
\end{eqnarray}

The following lemma gives a feasible solution to ALP such that cost is bounded by LP optimal within a constant factor.
\begin{lemma}
\label{fs_kflo}
A feasible solution $w'$ can be obtained to the  ALP such that $CostALP(w') \leq  2 \opt{LP}$.
\end{lemma}

\begin{proof}
Let $w'_i = y^*_{i}$.
\begin{enumerate}
    \item {\em Feasibility:} Constraints \ref{LPALP_const1} and \ref{LPALP_const2} hold vacuously as $\cfull$ and hence $\cstar$ are empty. 
    Constraints  \ref{LPALP_const3}, \ref{LPALP_const0} and \ref{LPALP_const4} hold by properties~\ref{prop1},~\ref{prop2} and~\ref{prop3} respectively.
    
    \item {\em Cost Bound:} 
    As $\T{j}=\starj$,
    $CostALP(w'_i) = \sum_{j \in \cliset} \sum_{i \in \starj} \distd{i}{j} y^*_{i} + \sum_{i \in \facilityset} f_i y^*_i \leq 2 \sum_{j \in \cliset} \sum_{i \in \starj} \dist{i}{j} x^*_{ij} + \sum_{i \in \facilityset} f_i y^*_i = 2\sum_{j \in \cliset} \sum_{i \in \facilityset} \dist{i}{j} x^*_{ij} + \sum_{i \in \facilityset} f_i y^*_i = 2\opt{}$. The inequality follows as  $c'(i,j) \le 2c(i,j)$ and $x^*_{ij} = y^*_i$.
\end{enumerate}
\end{proof} \qed
    
\subsection{Iterative Rounding}
We next present an iterative rounding algorithm(IRA) for solving the ALP. In every iteration of IRA, we compute an extreme point solution $w^*$ to ALP and check whether any of the constraints \ref{LPALP_const2} or \ref{LPALP_const3} has become tight. If a constraint corresponding to $j \in \cpart$ gets tight, we move the client to $\cfull$ and remove it from $\cpart$. We also update $\cstar$ so that disjointness property is satisfied. If a constraint corresponding to $\cli{f} \in \cfull$ gets tight, we shrink $\T{\cli{f}}$ to $\ball{\cli{f}}$;
update $\ball{\cli{f}}$ and $\cstar$ accordingly.  The algorithm is formally stated in Algorithm \ref{IRA}. For $j \in \cfull$, let $resp(j)$  be the client $j' \in \cstar$ who takes the responsibility of getting $j$ served. Whenever $j$ is added to $\cstar, resp(j)$ is set to $j$ and whenever it is removed because of another client $j' \in\cstar$, $resp(j)$ is set to $j'$. If $j$ was never added to $\cstar$, then there must be a $j'$ because of which it was not added to $\cstar$ in lines $4$ and $5$. Such a $j'$ takes the responsibility of $j$ in that case. Note that a client $j$ may be added and removed several times from $C^*$ over the iterations of the algorithm as $\T{j}$ and $\ball{j}$ shrink (see Fig.~\ref{fig_ALP} for illustration).

\begin{algorithm} 
	\footnotesize
	\begin{algorithmic}[1]
		\STATE $\cfull \leftarrow \phi$, $\cpart \leftarrow \cliset$, $\cstar \leftarrow \phi$, $\T{j}=\starj$, $\rtj{} = \dlj{}$\\
		\WHILE {true}
		\STATE Find an extreme point solution $w^*$ to ALP
		\IF{there exists some $j \in \cpart$ such that $\sum_{i \in \T{j}} w^*_i$ = 1} 
		{\STATE $\cpart \leftarrow \cpart \setminus \{j\}, \cfull \leftarrow \cfull \cup \{ j\}, \ballj \leftarrow \{ i \in \T{j} : \distd{i}{j} \leq \floor{\rtj{}/2}\}$
		\STATE $process-\cstar(j)$.}
		\ENDIF
		\IF{there exists $j \in \cfull$ such that $\sum_{i \in \ballj} w^*_i=1$}
		{\STATE $\T{j} \leftarrow \ballj, \rtj{} = \floor{\rtj{} /2}, \ballj \leftarrow \{ i \in \T{j} : \distd{i}{j} \leq \floor{\rtj{}/2} \}$ \STATE $process-\cstar(j)$} 
		\ENDIF
		\ENDWHILE 
	\STATE Return $w^*$
	\vspace{0.5cm}
	\STATE $process-C^*(j)$
        \IF{there exists $j' \in \cstar$ with $\rtj{'} < \rtj{}$ and $\T{j} \cap \T{j'} \neq \phi$} {\STATE $resp(j) = j'$, if there are more than one such $j'$s, choose any arbitrarily.}
	\ELSE {\STATE \textbf{if} $j \in \cstar$ \textbf{then} update $\rtj{}$ to its new value} \STATE \ \ \ \ \ \ \ \ \ \ \ \ \ \textbf{else} Add $j$ to $\cstar$ with $\rtj{}$ and $resp(j) = j$. 
{\STATE Remove  all $j'$ from $\cstar$ for which  $\rtj{} < \rtj{'}$ and $\T{j} \cap \T{j'} \neq \phi$, $resp(j') = j$. }
    \ENDIF
	\end{algorithmic}
	\caption{Iterative Rounding Algorithm}
	\label{IRA}
\end{algorithm}

\begin{figure}
	\begin{tabular}{ccc}
	\includegraphics[width=3.5cm]{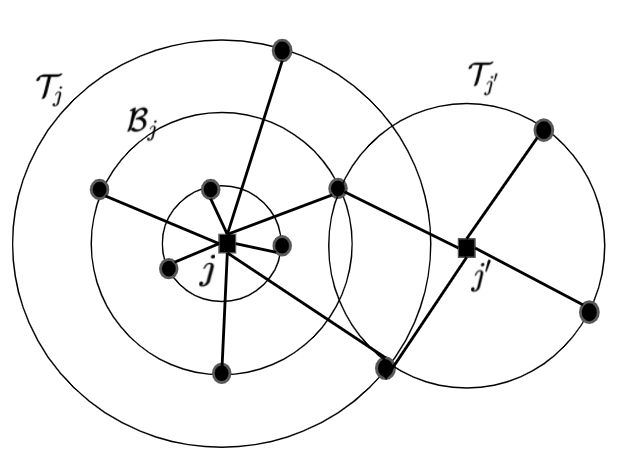}
		&
		\includegraphics[width=3.5cm]{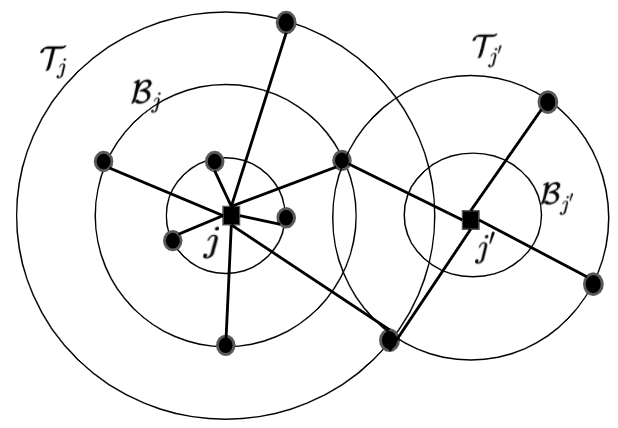}
		&
		\includegraphics[width=3.5cm]{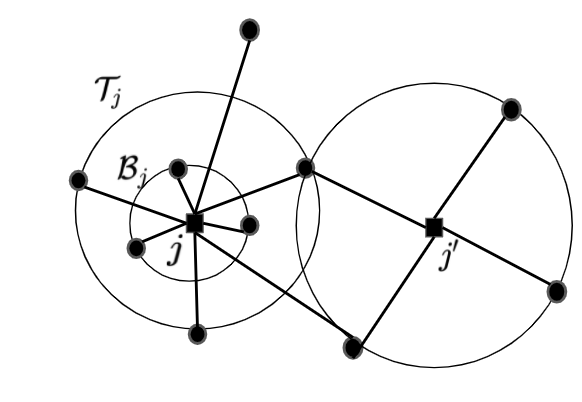}
		\\
		 ($a$)& \ \ ($b$)& \ \ ($c$) \ \
	\end{tabular}
	\caption{($a$) Initially both $j$ and $j'$ are in $\cpart$. Suppose $\sum_{i \in \T{j}} w^*_i = 1$, then $j$ is added to $\cfull$ and $\ballj$ is defined for $j$. $j$ is added to $C^*$ as well. ($b$) Subsequently, suppose $\sum_{i \in \T{j'}} w^*_i = 1$, then $j'$ is added to $\cfull$ and $\ball{j'}$ is defined for $j'$. $j'$ is added to $C^*$ whereas $j$ is removed from $C^*$ because $\rtj{'} < \rtj{}$ and $\T{j} \cap \T{j'} \neq \phi$ ($c$) Next, suppose $\sum_{i \in \ball{j}} w^*_i = 1$ in a future iteration, then $\T{j}$ and $\ballj$ shrink. $j$ is added to $C^*$ again and $j'$ is removed from $C^*$ because after shrinking $\rtj{}<\rtj{'}$.} 
	\label{fig_ALP}
\end{figure}

Lemmas~\ref{FS_ALP},~\ref{pseudo_integral}, and~\ref{neartocfull} help us analyse our algorithm. Lemma~\ref{FS_ALP} shows that the solution obtained in an iteration is feasible for the ALP of the next iteration.
We also prove that the cost of the solutions computed is non-increasing over iterations. 

\begin{lemma}
\label{FS_ALP}
Let $ALP_{t}$ and $ALP_{t+1}$ be the auxiliary LPs before and after iteration $t$ of IRA. Let $w^t$ be the extreme point solution obtained in $t^{th}$ iteration. Then
$(i)$ $w^t$ is a feasible solution to $ALP_{t+1}$, $(ii)$ $CostALP_{t+1}(w^{t}) \leq CostALP_{t}(w^{t})$ and hence $CostALP_{t+1}(w^{t+1}) \leq CostALP_{t}(w^{t})$.
\end{lemma}

\begin{proof}

Note that the feasibility and the cost can change only when one of constraints~(\ref{LPALP_const2}) or constraints~(\ref{LPALP_const3}) 
$w^t$ becomes tight, that is, either condition at step 4 or condition at step 6 of the algorithm is true.


\begin{itemize}
    \item[($i$)] 
    When one of constraints~(\ref{LPALP_const3}) corresponding to a client $j$ becomes tight i.e. $\sum_{i \in \T{j}} w^t =1$, we move client $j$ from $\cpart$ to $\cfull$ and
    define the set $\ballj$. Thus,
  $\sum_{i \in \ballj } w^t \leq \sum_{i \in \T{j} } w^t = 1$. Thus the new constraints added in constraints~\ref{LPALP_const2} and~\ref{LPALP_const1} (if $j$ is added to $\cstar$) are satisfied. Constraint~(\ref{LPALP_const4}) holds as $|\cfull|$ increases by $1$ and $\sum_{j \in \cpart} \sum_{i \in \T{j}} w^t$ decreases by $1$.  There is no change in constraint~\ref{LPALP_const0}.
    
    Let one of the constraints~(\ref{LPALP_const2}) corresponding to a full client $j$ becomes tight i.e. $\sum_{i \in \ballj} w^t =1$. Two things happen here: ($i$) we shrink $\T{j}$ to $\ballj$, hence $\sum_{i \in \T{j}} w^t=1$. Thus constraint~\ref{LPALP_const1} is satisfied if $j$  is added to $\cstar$. ($ii$) shrink $\ballj$ to half its radius, thus $\sum_{i \in \ballj } w^t \leq \sum_{i \in \T{j} } w^t = 1$. Thus constraint~\ref{LPALP_const2} corresponding to $j$ continue to be satisfied with the shrunk $\ballj$. 
    There is no change in constraints~\ref{LPALP_const0} and~\ref{LPALP_const4}.

    \item[($ii$)] For a client $j$, let $\dlj{}^{t}$, $\ballj^{t}$ and $\T{j}^{t}$ be the set $\dlj{}$, $\ballj$ and $\T{j}$ corresponding to client $j$ in $ALP_{t}$ and $\dlj{}^{t+1}$, $\ballj^{t+1}$ and $\T{j}^{t+1}$ be the respective values in $ALP_{t+1}$.
    
    \begin{enumerate}
    \item[a.] When $\T{j}$ and $\ballj$ shrink because constraint~\ref{LPALP_const2} becomes tight for a client $j$. Cost paid by $j$ in $w^t$ in the $t^{th}$ iteration
    $ = \sum_{i \in \ballj^{t}} c'(i,j) w^t $ because $\sum_{i \in \ballj^{t}} w^t=1$ in the $t^{th}$ iteration. Since  $\ballj^{t}= \T{j}^{t+1}$, $ \sum_{i \in \ballj^{t}} c'(i,j) w^t = \sum_{i \in \T{j}^{t+1} : c'(i,j) \leq \dlj{}^{t+1}/2} c'(i,j) w^t +  \sum_{i \in \T{j}^{t+1} : c'(i,j) = \dlj{}^{t+1}} c'(i,j) w^t = \\ \sum_{i \in \ballj^{t+1}} c'(i,j) w^t + (1- \sum_{i \in \ballj^{t+1}} w^t)\dlj{}^{t+1}=$ Cost paid by $j$ in $w^t$ in the $(t+1)^{th}$ iteration. Thus change in cost is $0$.
        
        \item[b.] When a client $j$ is moved from $\cpart$ to $\cfull$ because constraint~\ref{LPALP_const3} becomes tight.  Cost paid by $j$ in $w^t$ in the $t^{th}$ iteration $= \sum_{i \in \T{j}^{t}} c'(i,j) w^t  = \\ \sum_{i \in \T{j}^{t} : c'(i,j) \leq \rtj{}^{t}/2} c'(i,j) w^t + \sum_{i \in \T{j}^{t} : c'(i,j) = \rtj{}^{t}} c'(i,j) w^t = \\ \sum_{i \in \ballj^{t+1}} c'(i,j) w^t + (1-\sum_{i \in \ballj^{t+1}} w^t) \rtj{}^{t+1}=$ Cost paid by $j$ in $w^t$ in the $(t+1)^{th}$ iteration. Thus change in cost is $0$.
        
    \end{enumerate}
\end{itemize}
\end{proof} \qed

Thus we have, $CostALP_{t+1}(w^{t+1}) \le CostALP_{t+1}(w^{t}) = CostALP_{t}(w^{t})$ where the first inequality follows because $w^{t+1}$ is an extreme point solution and $w^t$ is a feasible solution to $ALP_{t+1}$. 
Hence, if $n$ is the number of iterations  of the IRA then $CostALP_n(w^n) \leq  CostALP_{1}(w^1) \leq CostALP(w') \le 2\opt{LP}$ where the second last inequality follows as $w^1$ is an extreme point solution and  $w'$ is a feasible solution for $ALP = ALP_1$,  last inequality follows from Lemma~\ref{fs_kflo}.
Let $w^*$ be the solution returned by the IRA, then $w^* = w^n$. Lemma~\ref{pseudo_integral} establishes that at the end of our IRA, solution $w^*$ is pseudo-integral.

\begin{lemma}
\label{pseudo_integral}
$w^*$ returned by Algorithm \ref{IRA} has at most two fractionally opened facilities.  
\end{lemma}
\begin{proof}
At the termination of the algorithm constraints~\ref{LPALP_const2} and~\ref{LPALP_const3} will not be tight. Let $n_f$ be the number of fractional variables at the end of the algorithm. Then there are exactly $n_f$ number of independent tight constraints from~(\ref{LPALP_const1}),~(\ref{LPALP_const0}) and~(\ref{LPALP_const4}). Let $X$ be the number of tight constraints of type~\ref{LPALP_const1}. There must be at least $2$ fractional variables corresponding to each of these constraints. Also, there must be at least $2$ fractional variables corresponding to constraint~\ref{LPALP_const1}, different from those obtained constraints~\ref{LPALP_const4}. Thus, $n_f \ge 2X + 2$ i.e. $X \le n_f/2 -1$. Also, the number of  tight constraints is at most $X + 2$ and hence is at most $n_f/2 + 1$ giving us $n_f \le n_f/2 + 1$ or $n_f \le 2$.
\end{proof} \qed

We open both the fractionally opened facilities at a loss of  $+ f_{max}$ in the facility opening cost where $f_{max}$ is the guess of the most expensive facility opened by the optimal.
In Lemma~\ref{neartocfull} we show that for a client $j$ in $\cfull \setminus \cstar$ there is some client in $ \cstar$, that is close to $j$, i.e. within $5\rtj{}$ distance of $j$.

\begin{lemma}
\label{neartocfull}
At the conclusion of the algorithm, for every $j \in \cfull$, there exists at least 1 unit of open facilities within distance $5 \rtj{}$ from j. Formally, $\sum_{i:c'(i,j) \leq 5\rtj{}} w^*_i \geq 1$.
\end{lemma}
\begin{proof}
 Let $j \in \cfull$.  
 If $resp(j) = j$, then this means that $j$ was added to $\cstar$ and was present in $\cstar$ at the end of the algorithm. Then, one unit is open in $\T{j}$ i.e. within a distance of $\rtj{}$ of $j$.
 
 If $resp(j) = j' (\ne j)$ then $j$ was either never added to $\cstar$ or was removed later. In either case responsibility of opening a facility in a close vicinity of $j$ was taken by $j'$. First we consider the case when $j$ was added to $\cstar$ but removed later. Let $j_0, j_1, \ldots j_r$ be the sequence of clients such that $resp(j_i) = j_{i-1}, i = 1 \ldots r, \ resp(j_0) = j_0$ and $j_r = j$. Since $resp(j_0) = j_0$, one unit is open in $\T{j_0}$ i.e. within a distance of $\rtj{0}$ of $j_0$. Clearly, $\rtj{i-1} \le \rtj{i}/2$. Thus $\rtj{i} \le (1/2)^{r -i} \rtj{r}$ for all $i = 0 \ldots r - 1$.  Thus, $\distd{j_{r}}{j_{0}} \le \sum_{i = 1~to~r} \distd{j_{i}}{j_{i-1}} \le \sum_{i = 1~to~r}  (\rtj{i} + \rtj{i - 1}) \le  \rtj{0} + 2\sum_{i = 1~to~r-1}  \rtj{i} + \rtj{r} \le \rtj{0} + 2\sum_{i = 1~to~r-1} (1/2)^{r -i} \rtj{r} + \rtj{r}$. Thus one unit of facility is open within a distance of $2\sum_{i = 0~to~r-1} (1/2)^{r -i} \rtj{r} \\ + \rtj{r} = \sum_{t = 1~to~r} (1/2)^{r - t} \rtj{r} + \rtj{r} \le 3 \rtj{r}$ from $j$. 
 
 Next, let $j$ was never added to $\cstar$. Then since $resp(j) = j' (\ne j)$, $j'$ was added to $\cstar$ at some point of time. Thus, from above one unit of facility is opened within distance $3\rtj{'}$ of $j'$. Also, $\distd{j}{j'} \le \rtj{} + \rtj{'} \le 2\rtj{}$. Thus, one unit of facility is opened within distance $5\rtj{}$ of $j$.
 
\end{proof} \qed

We run the algorithm for all the guesses of $f_{max}$ and select the one with the minimum cost.

\textbf{Combining Everything:}
Let $\bar{w}$ be our final solution. $Cost(\bar{w}) \leq Cost(w^*)+f_{max} \leq 5 \cdot CostALP_n(w^*)+f_{max} \leq 5\cdot2\cdot\opt{LP} + f_{max} \leq 11OPT$ where $OPT$ is the cost of the optimal solution.

\section{Tri-criteria for LBFLO}
\label{LBkFLO}
In this section, we present a tri-criteria solution for LBFLO problem with $\alpha=(0,1)$-factor violation in lower bound and at most $\beta=(\frac{1}{1-\alpha})$-factor violation in outliers at ($\lambda(\frac{1+\alpha}{1-\alpha})$)-factor loss in cost where $\lambda$ is approximation for FLO. Let $I$ be an instance of LBFLO. For a facility $i$, let $\N{i}$ be the set of $\lb{i}$  nearest clients. We construct an instance $I'$ of FLO with lower bounds ignored and facility costs updated as follows: if a facility $i$ is  opened in optimal solution of $I$, then it pays at least $\sum_{j \in \N{i}}c(i,j)$ cost for serving $\N{i}$ clients. Therefore, $f'(i)=f(i)+\delta \sum_{j \in \N{i}}c(i,j)$ where $\delta$ is a tunable parameter. 
\begin{lemma}
\label{fs}
Optimal solution of $I'$ is bounded by $ (\delta+1)Cost_I(O)$ where $O$ is the optimal solution of  $I$.
\end{lemma}

\begin{proof}
Clearly $O$ is a feasible solution for $I'$. 
Thus, service cost is same as that in $O$. 
And, 
$\sum_{i \in O} f'(i) = \sum_{i \in O}[f(i)+ \delta \sum_{j \in \N{i}}c(i,j)] \leq \delta Cost_I(O)$. 
   Therefore, $Cost_{I'}(I') \leq  (\delta+1)Cost_I(O)$.
\end{proof} \qed

Once we have an instance $I'$ of FLO, we use any algorithm for FLO to get a solution $AS'$ to $I'$ of cost no more than $\lambda Cost_{I'}(O')$ where $O'$ is the optimal solution of $I'$ and $\lambda$ is approximation solution for FLO. Note that a facility $i$ opened in solution $AS'$ might serve less than $\alpha\lb{i}$ clients as we ignored the lower bounds in instance $I'$. We close such facilities and do some reassignments to improve the violation in the lower bounds to $\alpha$;
in the process we make some violation in the number of outliers. 


We convert the solution $AS'$ to a solution $AS$ of LBFLO.
We close every facility $i$ that is serving less than $\alpha\lb{i}$ clients in $AS'$ and either reassign its clients to other opened facilities or decide to leave them unserved. Cost of reassignment is charged to the facility opening costs of the closed facilities.
Consider a facility $i$ opened in $AS'$ that served less than $\alpha \lb{i}$ clients. Let $\iclients$ be the set of clients, in $\N{i}$, assigned to $i$ in $AS'$ and $\bar \iclients$ be the remaining clients in $\N{i}$. 
 Since $i$ serves $<\alpha \lb{i}$ clients, $|\bar \iclients| \ge (1-\alpha )\lb{i}$.
Some of the clients in $\bar \iclients$ are outliers in $AS'$ and some are assigned to other facilities.
 Let $\Out{i}$ be the clients in $\N{i}$ that are outliers and $\R{i}$ be the clients in $\N{i}$ assigned to some other facilities. See Fig.~\ref{division of clients}($a$). 
  If $\R{i} \ne \phi$ then let $j' \in \R{i}$ be the nearest client to $i$. then, 
    \begin{equation}
    \label{eq1}
        c(i,j') \leq \frac{\sum_{j \in \R{i}} c(i,j)}{|\R{i}|} \leq \frac{\sum_{j \in \N{i}} c(i,j)}{|\R{i}|}
    \end{equation}
    
\begin{figure}[t]

	\begin{tabular}{cc}
		\includegraphics[width=5cm]{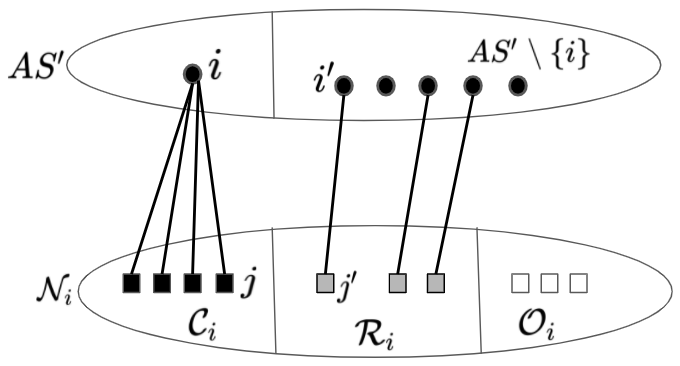}
		& 
		\includegraphics[width=5cm]{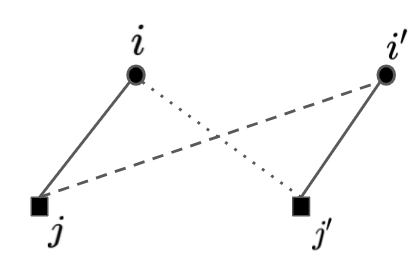}

		\\
		($a$)& \ \ ($b$) \ \
	\end{tabular}
	\caption{ ($a$) Division of clients in $\N{i}$ for a facility $i$ opened in $AS'$
		($b$) $ c(j,i') \leq c(i,j) + c(i,j') + c(j',i')$
	}
	\label{division of clients}
\end{figure}

Clients in $\iclients$ are  assigned to the facilities serving the clients in $\R{i}$ and are made outliers proportionally. That is, we assign $\frac{|\R{i}|}{|\R{i}|+|\Out{i}|} |\iclients|$ clients in  $\iclients$ to the nearest facility $i' \ne i$ opened in $AS'$ and leave  $\frac{|\Out{i}|}{|\R{i}|+|\Out{i}|} |\iclients|$ clients in  $\iclients$ unserved. If $|\R{i}| \neq 0$, then the total cost of reassignment is $\sum_{j \in \iclients} c(j,i') \leq \sum_{j \in \iclients} (c(i,j) + c(i,j') + c(j',i')) \text{(by triangle inequality, see Fig.~\ref{division of clients}($b$))}$ 
    $\leq \sum_{j \in \iclients}c(i,j) + ( \frac{|\R{i}|}{|\R{i}|+|\Out{i}|} |\iclients| \cdot 2 c(i,j'))$ (as $j'$ was assigned to $i'$ and not to $i$ in $AS'$)
    $\leq \sum_{j \in \iclients}c(i,j) + ( \frac{2|\iclients|}{|\R{i}|+|\Out{i}|} \cdot \sum_{i \in \N{i}} c(i,j))$ (using~(\ref{eq1}))
    $\leq \sum_{j \in \iclients}c(i,j) + ( \frac{2\alpha\lb{i}}{(1-\alpha)\lb{i}} \sum_{i \in \N{i}} c(i,j))$ (As $|\iclients| \leq \alpha\lb{i}$ and $|\R{i}|+|\Out{i}| \geq (1-\alpha)\lb{i}$)
    $\leq \sum_{j \in \iclients}c(i,j) + f'(i)$ (for $\delta \geq \frac{2\alpha}{1-\alpha}$).
    Thus the additional cost of reassignment of clients in $\iclients$ is bounded by the facility opening cost of $i$.
    Violation in outliers is $\frac{|\Out{i}| + \frac{|\Out{i}|}{|\R{i}|+|\Out{i}|} |\iclients|}{|\Out{i}|} \leq 1+\frac{|\iclients|}{|\R{i}|+|\Out{i}|} \leq 1+\frac{\alpha}{1-\alpha} = \frac{1}{1-\alpha}$. 

\textbf{Overall Cost Bound:}
It is easy to see that $Cost_I(AS)=Cost_{I'}(AS')$ as cost of solution $AS$ is sum of $(i)$ the original connection cost which is equal to the connection cost of $AS'$, ($ii$) the additional cost of reassignment, which is paid in $AS'$ by facilities that are closed in $AS$ and, ($iii$) the facility cost of the remaining facilities. 
%
%
Thus, $Cost_I(AS)=Cost_{I'}(AS') \leq \lambda Cost_{I'}(O') \leq \lambda (1+\delta)Cost_{I}(O) =
\lambda (\frac{1+\alpha}{1-\alpha})Cost_{I}(O)$ for $\delta=\frac{2\alpha}{1-\alpha}$.  
Using $\lambda=(3+\epsilon)$-approximation of Charikar \etal~\cite{charikar2001algorithms} for FLO, we get $(3+\epsilon)(\frac{1+\alpha}{1-\alpha})$ factor loss in cost for $\epsilon>0$. 

\section{Conclusion and Future Scope}
\label{conclusion}
In this paper, we first presented a $11$-factor approximation for $k$-facility location problem with outliers opening at most $k+1$ facilities. This also gives us the first constant factor approximation for FLO using LP rounding techniques. Our result can be extended to knapsack median problem with outliers with $(1+\epsilon)$ violation in budget using enumeration techniques.

We also gave a tri-crtieria, $(\alpha, \frac{1}{1-\alpha},(3+\epsilon)\frac{1+\alpha}{1-\alpha})$-solution for general LBFLO where $\alpha=(0,1)$ and $\epsilon>0$. It will be interesting and challenging to see if we can reduce the violation in outliers to $<2$ maintaining $\alpha >1/2$. 

We believe that using pre-processing and strengthened LP techniques of Krishnaswamy \etal~\cite{krishnaswamy-kMO} we can get rid of the $+1$ violation in cardinality for $k$FLO. This will also directly extend our tri-criteria solution to lower bounded $k$-facility location problem with outliers (LB$k$FLO). 




\bibliographystyle{splncs04}
\bibliography{ref}

\end{document}